\documentclass[smallextended]{svjour3}       

\usepackage{etex} 
\usepackage{url,amssymb, bm, mathrsfs,amsmath}

\usepackage{fullpage}

\usepackage{graphicx,psfrag}
\usepackage[pdf]{pstricks}  
\graphicspath{{./figures/},{./20170831-version/}}

\usepackage[colorlinks=true,linkcolor=blue,backref=page]{hyperref}
\renewcommand*{\backref}[1]{}

\newcommand{\st}{\mathop{\rm ~s.t.~}}

\title{Supervisory Control of Discrete-event Systems under
  Attacks\thanks{This work was supported by the JSPS KAKENHI Grant
    Number JP17K14699, the National Science Foundation award
    No.~1136174, and the U.S.~Office of Naval Research under MURI
    grant No.~N00014-16-1-2710.}}

\author{
  Masashi Wakaiki
  \and
  Paulo Tabuada
  \and
  Jo\~ao P. Hespanha
}

\date{October 30, 2017}

\institute{ %
  Masashi Wakaiki \at Graduate School of System Informatics, Kobe
  University, Kobe, 657-8501, Japan,
  \email{wakaiki@ruby.kobe-u.ac.jp}.
  \and
  Paulo Tabuada
  \at Department of Electrical Engineering, University of California, Los
  Angeles, CA 90095-1594, USA, \email{tabuada@ee.ucla.edu}.
  \and
  Jo\~ao P. Hespanha
  \at Center for Control, Dynamical-systems, and Computation, University
  of California, Santa Barbara, CA 93106-2560, USA,
  \email{hespanha@ucsb.edu}.
}

\begin{document}
\maketitle

    









\begin{abstract}
  We consider a multi-adversary version of the supervisory control
  problem for discrete-event systems, in which an adversary corrupts
  the observations available to the supervisor. The supervisor's goal
  is to enforce a specific language in spite of the opponent's actions
  and without knowing which adversary it is playing against. This
  problem is motivated by applications to computer security in which a
  cyber defense system must make decisions based on reports from
  sensors that may have been tampered with by an attacker. We start by
  showing that the problem has a solution if and only if the desired
  language is controllable (in the Discrete event system classical
  sense) and observable in a (novel) sense that takes the adversaries
  into account. For the particular case of attacks that
    insert symbols into or remove symbols from the sequence of sensor
    outputs, we show that testing the existence of a supervisor and
    building the supervisor can be done using tools developed for the
    classical DES supervisory control problem, by considering a family
    of automata with modified output maps, but without expanding the
    size of the state space and without incurring on exponential
    complexity on the number of attacks considered.
\end{abstract}


\section{Introduction}

\emph{Discrete event systems} (DESs) are non-deterministic transition
systems defined over a typically finite state-space. The DESs
supervisory control problem refers to the design of a feedback
controller --- called a \emph{supervisor} --- that restricts the set
of possible sequences of transitions (typically represented by
\emph{strings} over an alphabet of transitions) to a desired set $K$.
The supervisor's task is complicated by the fact that (i) only a
subset of transitions can be inhibited (the so called ``controllable''
transitions) and (ii) the supervisor only has partial information
about the state of the system, which it gathers by observing a string
of ``output symbols.'' This basic problem is motivated by a wide range
of applications that include manufacturing systems, chemical batch
plants, power grids, transportation systems, database management,
communication protocols, logistics, and computer security. The latter
is the key motivating application for the work reported here.


\medskip

We consider a zero-sum multi-adversary
version of the supervisory control problem, where one player (the
supervisor) faces multiple adversaries with distinct action spaces and
the supervisor needs to find a policy that wins against any of its
possible adversaries, without knowing which one is the actual
opponent.
The supervisor loses if it accepts strings
outside the desired set $K$ or if it rejects strings within
$K$.
We consider a zero-sum game in which the adversary tries to make
the supervisor lose by manipulating the string of output symbols that
the supervisor uses to make decisions (see
Figure~\ref{fig:closed_loop_under_attack}).

\medskip

In computer security applications, the supervisor is typically
responsible for enacting defense mechanisms, such as opening/closing
firewalls, starting and stopping services, authorizing/deauthorizing
users, and killing processes. Such decisions are based on observations
collected by cyber-security sensors that log events like user
authentication, network traffic, email activity, and access to
services or files. Ultimately, the supervisors' goal is to allow
users to perform all the tasks they are authorized to do, while
preventing unauthorized access to resources and services.
%
%
The zero-sum multi-adversary game is motivated by scenarios where a
cyber attacker manipulates the supervisor's observations by tampering
with one or more of the cyber-security sensors. A key challenge is that
security mechanisms typically do not know if a particular sensor has
been compromised so it needs to consider multiple alternatives for
sensor manipulation, which explains the need for the multi-adversary
model.

\begin{figure}[tb]
  \centering
  \includegraphics[width = 9cm]
  {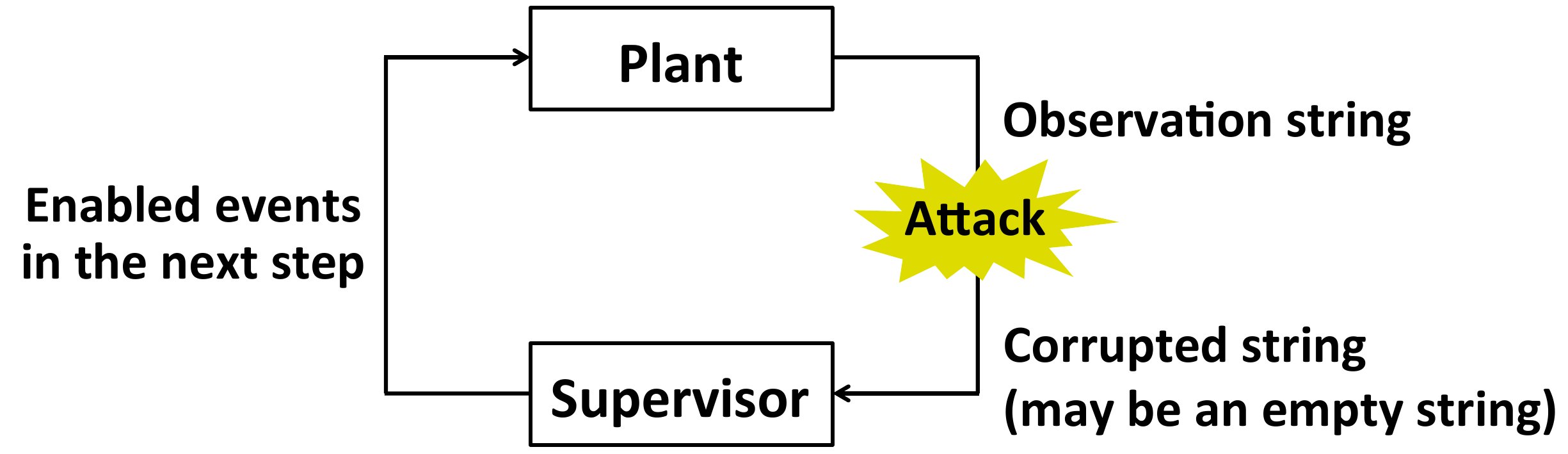}
  \caption{Closed-loop system under attacks.}
  \label{fig:closed_loop_under_attack}
\end{figure}

\medskip

One of the paper's main results is a general necessary and sufficient
condition for the multi-adversary game to have a ``solution'', i.e.,
for the existence of a supervisor that can win the game against any of
the adversaries, without knowing which one it is playing against
(Section~\ref{se:existence}). This necessary and sufficient condition
is expressed in terms of a controllability condition, which also
arises in the classical supervisory control setup and a novel
observability condition that depends on how the adversaries can
manipulate/attack the output string. The result is constructive in the
sense that it provides an explicit formula for the supervisor, in case
one exists, but this general formula is typically hard to use in
practice.

\medskip

The second part of the paper is focused on ``output-symbol attacks'',
which correspond to games where each of the adversaries is restricted
to attack a specific set of symbols. Specifically, each adversary is
able to insert into or remove from the output string symbols from a
given set. In the context of computer security, we can view such
adversaries as attackers that infiltrated a security sensor and
insert/remove entries from the corresponding security log. For
output-symbol attacks, we show that testing observability and
generating supervisors can be done using tools developed for classical
supervisory control problems.

\medskip

For output-symbol attacks, we show that testing the new notion of
observability under attacks can be done by checking the classical
notion of observability for a family of DESs that differ from each
other only by their output maps
(Section~\ref{se:output-symbol-attacks-test}). These output maps are
obtained from the original output map by considering all pairs of
adversaries and \emph{removing} from the output string every symbol
that 2 adversaries could attack if they were allowed to combine
their efforts (which they are not). This shows that observability
under attacks can still be tested in polynomial time.

\medskip

Also for output-symbol attacks, we show that a supervisor for the
multi-adversary case can be constructed from a family of classical
supervisors, each constructed for one of the classical DES supervisory
control problems used for the observability test
(Section~\ref{se:output-symbol-attacks-supervisor}). Since
constructing classical supervisors has exponential worst-case
complexity on the number of states, this procedure also has
exponential worst-case complexity, but it has the desirable feature
that it does not enlarge the state-space of the original system. This
is in sharp contrasts with the alternative approach of reducing the
multi-adversary supervisory control problem to a classical supervisory
control problem, which generally requires expanding the state-space
$M$ times, where $M$ is the number of adversaries.

\medskip

Supervisory control theory provides a natural framework to address
plant uncertainties and faults, either in a robust control context
(e.g., \cite{Lin1993,Takai2000,Saboori2006}) or in a fault tolerant
context (e.g., \cite{Sanchez2006,Paoli2011,Shu2014}). Several aspects
of security have also been considered in the DES literature, including
the stealthy deception attacks in \cite{Amin2013,Teixeira2015} that
aim at injecting false information without being detected by the
controller; or the opacity of DESs in
\cite{Takai2008,Dubreil2010,Saboori2012,Wu2014} and references there
in, whose goal is to keep the system's secret behavior uncertain to
outsiders. Intrusion detection in the DES framework has been
investigated in \cite{Thorsley2006,Whittaker2008,Hubballi2011}, with
the goal of guaranteeing the confidentiality and integrity of DESs.

\medskip

The key novelty of the problem considered here is the focus on DES
problems against adversaries that manipulate the observations available
to the supervisor. By focusing on this type of attack, we obtain tests
for the existence of a supervisor and algorithms to generate the
supervisor that do not require an expansion of the state-space of the
original (not attacked) DES.

\medskip

DESs with non-deterministic (set-valued) observation maps have been
studied in \cite{Xu2009,Ushio2016,yin2017supervisor}, where the
problem formulation could be viewed has a version of the problem
considered here but with a single adversary.
%
The authors in \cite{Xu2009} introduced the notion of lifting and
synthesized supervisors by converting the original problem with
nondeterministic observations to a problem in the lifted domain with
deterministic observations.  In \cite{Ushio2016,yin2017supervisor},
the so-called Mealy automata were used to deal with both
state-dependent observations and non-deterministic outputs.
The work reported here is inspired by the research on state estimation
under sensor attacks developed in
\cite{Fawzi2014,Shoukry2016,Chong2015}, where the authors consider a
linear time invariant system with multiple outputs and an attacker
that can manipulate a subset of the sensors that measure these
outputs.
The goal is to estimate the system's state and stabilize it using
feedback control, without knowing which sensors have been compromised.

\section{Background on supervisory control of discrete-event systems}\label{se:background}

We start by recalling basic notation and definitions that are common
in the Discrete-Event Systems (DES) literature.

\medskip
  
For a finite set $\Sigma$ of \emph{events,} we denote by $|\Sigma|$ the number
of elements of $\Sigma$ and by $\Sigma^*$ the set of all finite strings of
elements of $\Sigma$, including the empty string $\epsilon$. A subset $L$ of
$\Sigma^*$ is called a \emph{language over the alphabet $\Sigma$} and the
\emph{prefix closure} of $L$ is the language
\begin{align*}
  \bar L := \big\{u\in\Sigma^*:\exists v\in\Sigma^*,~u v\in L\big\},
\end{align*}
where $u v$ denotes the concatenation of two strings in $\Sigma^*$. The
language $L$ is said to be \emph{prefix closed} if $L=\bar L$.  We
define the concatenation of languages $L_1, L_2 \subset S^{*}$ by
\begin{align*}
  L_1L_2 :=
  \{
  w_1w_2 \in \Sigma^*:
  w_1 \in L_1,~w_2 \in L_2
  \}.
\end{align*}
Consider an \emph{automaton} $G = (X, \Sigma, \xi, x_0)$, where
$X$ is the set of states, $\Sigma$ the nonempty finite set of events,
$\xi:X\times \Sigma \to X$ the transition mapping (a partially defined function), and
$x_0 \in X$ the initial state.  We write
$\xi(x,\sigma)!$ to mean that $\xi(x,\sigma)$ is defined.  The transition function
$\xi$ can be extended to a function $X\times \Sigma^* \to X$ according to the
following inductive rules
\begin{align*}
  \xi(x,\epsilon) &:= x, & \forall& x \in X\\
  \xi(x,w\sigma)&:=
  \begin{cases}
    \xi(\xi(x,w), \sigma) &
    \text{if $\xi(x,w)!$ and
      $\xi(\xi(x,w), \sigma)!$} \\
    \text{undefined} & \text{otherwise},
  \end{cases}
  & \forall& x \in X,\; w \in \Sigma^*, \;\sigma \in \Sigma.
\end{align*}
The \emph{language generated by $G$} is then defined by
\begin{align*}
  L(G) := \{ w \in \Sigma^*: \xi(x_0,w)!\}.
\end{align*}


Let $\Sigma$ be a set of events that can be partitioned in two disjoint
sets as $\Sigma=\Sigma_c\cup\Sigma_{uc}$, where $\Sigma_c$ is called the \emph{set of
  controlled events} and $\Sigma_{uc}$ the \emph{set of uncontrolled
  events}.  For a language $L$ defined on $\Sigma$, a prefix-closed set
$K \subset L$ is said to be \emph{controllable} if $ K \Sigma_{uc} \cap L \subset K.  $

\medskip

Consider an \emph{observation map}
$P:\Sigma\to(\Delta\cup\{\epsilon\})$ that maps the set of events
$\Sigma$ into a set of \emph{observations $\Delta$} (augmented by the empty
event $\epsilon$). This observation map $P$ can be extended to the map
defined for strings of events according to $P(\epsilon) = \epsilon$ and
\begin{align*}
  P(w\sigma)=P(w)P(\sigma), \quad
  \forall w\in\Sigma^*,~\sigma \in \Sigma.
\end{align*}
A prefix-closed language $K\subset L$ is \emph{$P$-observable with respect
  to $L$} if
\begin{align*}
  \ker P \subset act_{K\subset L}
\end{align*}
where $\ker P$ denotes the relation on $\Sigma^*$ defined by
\begin{align*}
  \ker P&:= \big\{(w,\bar w)\in\Sigma^*\times\Sigma^*:P(w)=P(\bar w)\big\}
\end{align*}
and $act_{K\subset L}$
the 
binary relation on $\Sigma^*$ defined by
\begin{align*}
  act_{K\subset L}
  :=\Big\{ (w,\bar w)\in\Sigma^*\times\Sigma^*&:w,\bar w\in K \\
  &\quad \Rightarrow \quad
  \nexists\sigma\in\Sigma \st
  [w\sigma\in K, \bar w\sigma\in L\setminus K] \text{ or }
  [w\sigma\in L\setminus K, \bar w\sigma\in K]
  \Big\}.
\end{align*}
The key control design problem in DESs is to design a ``supervisor''
that modifies the original language $L$ generated by the automaton to
a desired language $K\subset L$, by judiciously disabling controllable
events in $\Sigma_c$ (and their associated transitions) based on the
observations obtained through $P$, in a feedback fashion. The reader
is referred to \cite{Ramadge1989,CassandrasLafortune2008,Wonham_Lec} for more
details on DES models and key results.

%
%

\section{Supervised discrete-event systems under attacks}\label{se:general}

In this paper, we deviate from the classical DES model by introducing
a set of adversaries $\mathcal{A}$ whose goal is to prevent the supervisor
from achieving the desired language $K$.
Each adversary $A\in\mathcal{A}$ can corrupt the string of output symbols
$P(w)$, $w\in\Sigma^*$ in multiple (non-deterministic) ways, e.g., erasing
and/or inserting specific output symbols (see
Example~\ref{ex:security} below). Our goal is thus to design a
supervisor that can guarantee the desired language $K$ regardless of
(1) which $A\in\mathcal{A}$ is the actual attack, and (2) how the actual
attack $A\in\mathcal{A}$ corrupts each string of output symbols $P(w)$,
$w\in\Sigma^*$.


\medskip

Formally, each adversary is a set-valued map
$A:\Delta^*\to2^{\Delta^*}$ that assigns to each string of output symbols
$P(w)$, $w\in\Sigma^*$ a set $A\big(P(w)\big)$ of (possibly distorted)
strings of symbols that the adversary can send to the supervisor,
instead of the original string $P(w)$. A set-valued map is convenient
because it enable us to consider multiple ways by which the attack $A$
may corrupt a particular string $P(w)$. We call the map $A$ an
\emph{observation attack} and the map $A P:\Sigma^*\to2^{\Delta^*}$ obtained from
the composition $A P := A\circ P$ the corresponding \emph{attacked
  observation map}. The attack map
$A_{\rm id}:\Delta^*\to 2^{\Delta^*}$ that assigns to each string
$y\in\Delta^*$ the set $\{y\}$ containing only the original output string
$y$ can be viewed as the absence of an attack.

\medskip



A \emph{$P$-supervisor for a language $L \subset \Sigma^*$ and an
  attack set $\mathcal{A}$} is a function
$f:\bigcup_{A \in \mathcal{A}}AP(L)\to 2^\Sigma$. The supervisor is \emph{valid
  (for the set of uncontrollable events $\Sigma_{uc}$)} if
\begin{align*}
  f(y) \in \Gamma, \quad \forall y \in \bigcup_{A \in \mathcal{A}}AP(L),
\end{align*}
where
$\Gamma := \{ \gamma \subset \Sigma: \Sigma_{uc} \subset \gamma\}$ is the set of all subsets of
$\Sigma$ that contain all the uncontrollable symbols in $\Sigma_{uc}$. One should
view $f(y)$ as the set of events that the supervisor enables after
observing the (potentially distorted) string $y\in\Delta^*$. A valid
supervisor is forced to always enable the uncontrolled symbols in
$\Sigma_{uc}$.

\medskip

By disabling symbols, a supervisor $f$ effectively only ``accepts'' a
subset of the strings in the original language $L$. However, which
symbols are accepted depends on how the adversary actually distorts
the observation strings. For a given supervisor $f$, the \emph{maximal
  language $L_{f,A}^{\max}$ controlled by $f$ under the attack
  $A \in \mathcal{A}$} is defined inductively by
\begin{align*}
  \begin{cases}
    \epsilon\in L_{f,A}^{\max}\\
    w\sigma\in L_{f,A}^{\max}
    \quad \Leftrightarrow \quad
    w\in L_{f,A}^{\max},\quad
    w\sigma\in L,
    \quad \exists y\in A P(w)
    \st \sigma\in f(y),
  \end{cases}
\end{align*}
whereas the \emph{minimal language $L_{f,A}^{\min}$ (
  $\subset L_{f,A}^{\max}$ ) controlled by $f$ under the attack
  $A\in \mathcal{A}$} is defined inductively by
\begin{align*}
  \begin{cases}
    \epsilon\in L_{f,A}^{\min}\\
    w\sigma\in L_{f,A}^{\min}
    \quad \Leftrightarrow \quad
    w\in L_{f,A}^{\min},\quad w\sigma\in L,\quad
    \forall y\in A P(w),\; \sigma\in f(y).
  \end{cases}
\end{align*}
By construction $L_{f,A}^{\max}$ and $L_{f,A}^{\min}$ are prefix
closed.  One can conclude from these definitions that an adversary
$A\in\mathcal{A}$ can distort the observations so that every string in the
maximal language $L_{f,A}^{\max}$ is accepted by the supervisor $f$,
but no string outside this language will be accepted. The same
adversary is also able to cause the rejection of every string outside
the minimal language $L_{f,A}^{\min}$, but is unable to cause the
rejection of any string inside this language.

\medskip

In the absence of an attack (i.e., when $A=A_{\rm id}$) the maximal
and minimal languages coincide and correspond to the classical notion
of the \emph{language $L_f$ generated by the supervised
  system}. However, under attacks we can only be sure that
$L_{f,A}^{\min}\subset L_{f,A}^{\max}$, typically with a strict
inclusion. Motivated by the desire to construct supervisors that
guarantee a specific language $K\subset L$, we say that $f$ is a
\emph{solution to the supervision of $K\subset L$ under the attack set
  $\mathcal{A}$} if $f$ is valid and
\begin{align*}
  L_{f,A}^{\min}=L_{f,A}^{\max}=K, \quad \forall A\in\mathcal{A}.
\end{align*}
This means that such $f$ is a Stackelberg equilibrium policy for the
supervisor (which we regard as the leader) and guarantees victory
because, even if the supervisor advertises $f$ as its policy, the
adversary cannot force the rejection of any string in $K$ (because
$L_{f,A}^{\min}=K$) and cannot force  the acceptance of any string
outside $K$ (because $L_{f,A}^{\max}=K$).

\begin{example}[Multi-layer cyber attack to a computer
  system]\label{ex:security}
  \begin{figure}[ht]
    \centering
    \includegraphics[width=.7\textwidth]{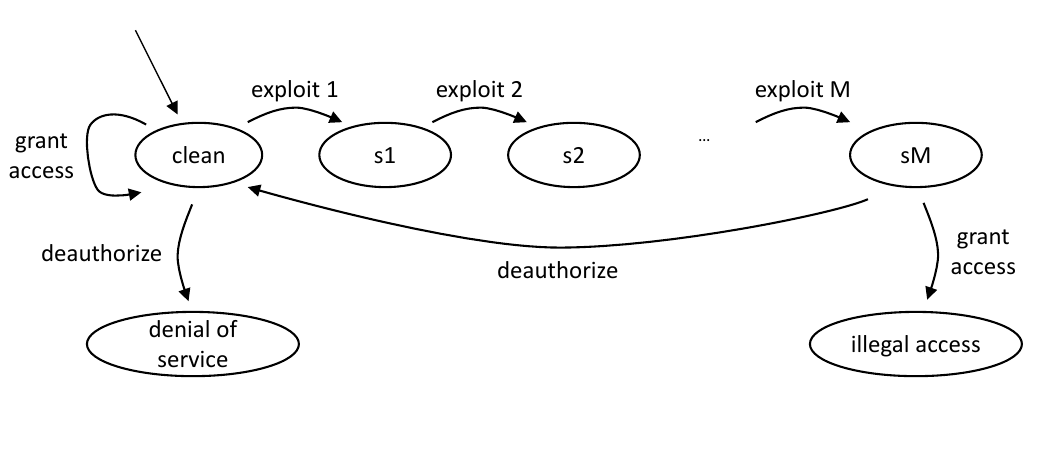}
    \caption{Automaton representation of a multi-layered cyber attack
      to a computer system. }
    \label{fig:multi-layer}
  \end{figure}
  Figure~\ref{fig:multi-layer} shows an automaton representation of a
  multi-layer cyber attack to a computer system. Each automaton's
  state represents the interaction between a specific user and the
  system. This interaction starts in a ``clean'' state where
  legitimate users can request access to a secure service by providing
  appropriate credentials. However, through a sequence of exploits, a
  cyber attacker may get access to credentials that would give her
  access to the secure service. The goal is to design a supervisor
  that grants access to the service to legitimate users but not to the
  cyber attacker that made use of the exploits to obtain the needed
  credentials. Figure~\ref{fig:multi-layer} shows a very simple
  sequential set of $M$ exploits, but realistic scenarios are
  characterized by much more complicate networks of exploits.

  \medskip

  In the example in Figure~\ref{fig:multi-layer}, the set of events is
  \begin{align*}
    \Sigma=\Big\{\text{grant acess},\; \text{deauthorize},\; \text{exploit 1},
    \text{exploit 2}, \dots,\text{exploit $M$} \Big\},
  \end{align*}
  where the ``exploits'' are uncontrolled events, whereas ``grant
  access'' and ``deauthorize'' as controlled events:
  \begin{align*}
    \Sigma_c&=\Big\{\text{grant acess},\; \text{deauthorize}\Big\}, &
    \Sigma_{uc}&=\Big\{\text{exploit 1},\;
    \text{exploit 2},\; \dots,\text{exploit $M$} \Big\}.
  \end{align*}
  We denote by $L\subset\Sigma^*$ the language $G$ represented by the automaton
  described above and by $K_{\rm safe}\subset L$ the language that excludes
  from $L$ all strings that include transitions to the ``bad states''
  labeled with ``denial of service'' and ``illegal access.''
  Specifically, $K_{\rm safe}$ does not include any string that starts
  with ``deauthorize'' nor any string that includes the sequences
  ``grant access, deauthorize''; ``deauthorize, deauthorize''; or
  ``exploit M, grant access.''

  \medskip

  While the ``exploits'' are uncontrolled events, we assume that the
  system has been instrumented with security logs that allow us to
  observe those transitions, which corresponds to the following
  observation map
  \begin{align}\label{eq:full-information}
    \Delta&=\Sigma, & P(w)&=w, \quad \forall w\in\Sigma^*.
  \end{align}
  In the absence of attacks (i.e., with $\mathcal{A}=\{A_{\rm id}\}$), the
  supervisor
  \begin{align*}
    f(w)&=
    \begin{cases}
      \Sigma\setminus  \text{``grant access''} & w\in \Sigma^*\, \Sigma_{uc}^+,\\
      \Sigma\setminus  \text{``deauthorize''} & \text{otherwise},
    \end{cases}
    &\Sigma_{uc}^+:=\Sigma_{uc}^*\setminus\{\epsilon\},
  \end{align*}
  disables ``grant access'' after strings that end with an exploit
  event and results in maximal and minimal languages that are both
  equal to the desired language: $L_{f,A}^{\max}=L_{f,A}^{\min}=K_{\rm safe}$.

  \medskip

  We are interested in situations where the security systems that
  generate the logs may have been compromised and create false and/or
  remove legitimate observations. For simplicity, suppose that the
  observations corresponding to each of the $M$ exploits are generated
  by $M$ security systems and we suspect that one of them may have
  been compromised. We can model this scenario by considering the
  following set of $M+1$ attacks:
  \begin{align}\label{eq:log-attacks}
    \mathcal{A}_{\rm log-attacks}:=\big\{A_{\rm id},A_\text{exploit 1},A_\text{exploit
      2},\dots,A_\text{exploit M} \big\},
  \end{align}
  where $A_{\rm id}$ corresponds to all security systems being good
  and each $A_\text{exploit\; i}$ corresponds to an attacker having
  compromised the security system that logs the occurrence of
  ``exploit i'', which may therefore arbitrarily insert/remove the
  output symbols ``exploit i'' into/from the (correct) observation
  string $P(w)$, $w\in\Delta^*$. The key question addressed in this paper is
  whether or not there exists a solution to the supervision of the
  language $K_{\rm safe}$ under this more interesting attack set
  $\mathcal{A}_{\rm log-attacks}$. We shall see later that this is
  possible for $M\ge 3$, but not possible for $M\in\{1,2\}$.

  \medskip

  The paper \cite{SheynerWing04} describes examples of the type of
  multi-layer attack considered in Example~\ref{ex:security}. One of
  these examples considers a scenario where an attacker wants to gain
  access to a database that requires root access at a particular Linux
  server. The server is behind a firewall and gaining access to it
  requires a sequence of exploits that include (1) exploiting a buffer
  overflow vulnerability in the Microsoft IIS Web Server to gain
  administrative privileges inside the firewall; (2) scanning the
  local network ports taking advantage of a misconfigured access
  control list in a Squid web proxy; (3) through this scan discovering
  and exploiting a vulnerability in a URL parsing function of the LICQ
  software to execute arbitrary code on a client's computer; and
  finally (4) letting a user without administrative privileges gain them
  illegitimately through a local buffer overflow. This example would
  correspond to a simple sequence of $M=4$ exploits in
  Figure~\ref{ex:security}, but \cite{SheynerWing04} discusses several
  alternative options to achieve the same goal by exploring these and
  other vulnerabilities. In fact, the examples considered in
  \cite[Section 4]{SheynerWing04} correspond to a network of exploits
  with over 300 nodes and many more transitions. The examples
  in~\cite{SheynerWing04} are quite specific in mapping exploits to
  known system vulnerabilities reported in the CVE database
  \cite{CVElist}. However, in practice, many system vulnerabilities
  are not known at the time cyber defense systems are designed so the
  type of analysis proposed in this paper should be based on automata
  that include exploits corresponding to \emph{unknown
    software/hardware vulnerabilities.}  E.g., one may include
  exploits like privilege escalation from user to administrative, even
  if no specific vulnerability to accomplish this is known for the
  system under consideration. Note that a system like
  \cite{ji2017rain} will be able to log the privilege escalation by a
  running process, even if the specific vulnerability that
  accomplishes this is unknown. The results in this paper seek to
  establish algorithms to process logs (like the ones generated by
  systems like \cite{ji2017rain}) that are robust to manipulation by
  an adversary.  \hspace*{\fill} $\Box$
\end{example}

\begin{remark}[Reduction to a supervisory control problem without
  attacks]\label{re:expansion}
  It is often possible to reduce the problem of finding a solution to
  the supervision of a language $K$ under an attack set $\mathcal{A}$ to a
  classical supervisory control problem without attacks. However, this
  is achieved by increasing the complexity of the original language
  $L$ and the automaton $G$ that generates it, which is problematic
  from a computational perspective.

  \medskip

  Such a reduction is easy to illustrate for the problem described in
  Example~\ref{ex:security}, which could also be modeled by a DES
  without attacks by expanding the state-space of the automaton
  $G$. To do this, we would replicate $M$ times the state of the
  automaton in Figure~\ref{fig:multi-layer}, each replicate
  corresponding to one output-symbol attack $A_\text{exploit\; i}$,
  $i\in\{1,2,\dots,m\}$; and add a new initial state with unobservable
  transitions to the ``clean'' state of every replicate. These
  transitions from the new initial state define which attack
  $A\in\mathcal{A}_{\rm log-attacks}$ will actually take place, but are not
  visible to the supervisor. The transitions within each one of the
  $M$ replicates also need to be adjusted: taking, e.g., the replicate
  corresponding to the symbol ``exploit 1'', to emulate the insertion
  of this symbol we would add self-transitions at every state with the
  event ``exploit 1'', and to emulate the removal of this symbol, we
  would associate the original transition from ``clean'' to ``s1''
  with a symbol that is not observable. This construction results in a
  new automaton $G_{\rm expanded}$ and an associated output map
  $P_{\rm expanded}$ with the property that any of its output strings
  $P_{\rm expanded}(w)$, $w\in L(G_{\rm expanded})$ could have been
  generated by the original automaton, under one of the attacks in
  $\mathcal{A}_{\rm log-attacks}$. However, $G_{\rm expanded}$
  has 
  $O(M^2)$
  states and
  transitions, as opposed to the
  $O(M)$ states and transitions in the original automaton $G$.

  \medskip

  The reduction of supervisory control problem under attacks to a
  classical problem without attacks may be more complicated than what
  we have seen in this example. However, it will generally require
  some form of state replication of the original automaton $|\mathcal{A}|$
  times, one for each possible attack map in the set $\mathcal{A}$, with
  initial unobservable transitions to each replicate. In addition, it
  will generally require the addition of further transitions within
  each replicate to model each of the attacks in $\mathcal{A}$.  \hspace*{\fill} $\Box$
\end{remark}

\section{Existence of supervisor achieving specifications}\label{se:existence}

A necessary and sufficient condition for the existence of a solution
to the supervision of $K$ under $\mathcal{A}$ can be expressed in terms of
a notion of observability that extends the conventional one presented
in Section~\ref{se:background}.  Given an attack set $\mathcal{A}$, we
say that a prefix-closed language $K\subset L$ is \emph{$P$-observable for
  the attack set $\mathcal{A}$} if
\begin{align}\label{eq:act-K-attacks-sets}
  R_{A,\bar A}  \subset act_{K\subset L},\qquad \forall A,\bar A\in\mathcal{A},
\end{align}
where the relation $R_{A,\bar A}$ contains all pairs of strings that may
result in attacked observation maps $AP$ and $\bar AP$ with a common
string of output symbols, i.e.,
\begin{align*}
  R_{A,\bar A}:= \big\{(w,\bar w)\in\Sigma^*\times\Sigma^*: A P(w)\cap \bar AP(\bar w)\neq \emptyset\big\}.
\end{align*}
The $P$-observability condition \eqref{eq:act-K-attacks-sets} can be equivalently
restated as requiring that,
$\forall w,\bar w\in K$,
\begin{align}
  \exists A,\bar A &\in \mathcal{A} \st
  AP(w) \cap \bar AP(\bar w) \neq \emptyset \notag \\
  &\quad
  \Rightarrow \quad
  \nexists
  \sigma \in \Sigma \st
  [w\sigma \in K,~
  \bar w\sigma \in L \setminus K] \text{~~or~~}
  [w\sigma \in L \setminus K,~
  \bar w\sigma \in K].
  \label{eq:observability_nonexist}
\end{align}
or equivalently
\begin{align}\label{eq:act-K-attacks}
  \exists A,\bar A \st A P(w)&\cap \bar AP(\bar w)\neq \emptyset \notag \\
  &\quad \Rightarrow \quad
  \forall\sigma\in\Sigma :
  w\sigma\notin L \text{ or }
  \bar w\sigma\notin L \text{ or }
  w\sigma,\bar w\sigma\in K \text{ or }
  w\sigma,\bar w\sigma\in L\setminus K.
\end{align}
In words, observability means that we cannot find two attacks
$A, \bar A \in \mathcal{A}$ that would result in the same observation
$y\in\Delta^*$ for two strings $w,\bar w \in K$ such that
$(w,w′) \notin act_{K\subset L}$, i.e., two $w, \bar w\in K$ such that one will
transition to an element in $K$ and the other to an element outside
$K$, by the concatenation of the same symbol $\sigma\in\Sigma$.

\medskip

The following theorem is the main result of this section. When the set
of attacks $\mathcal{A}$ contains only $A_{\rm id}$, it specializes to
the classical result of \cite{Ramadge1989}.
\begin{theorem}\label{th:observability-attacks}
  For every nonempty prefix-closed set $K\subset L$ and every attack set
  $\mathcal{A}$:
  \begin{enumerate}
  \item \label{en:observability-attacks}
    there exists a solution $f$ to the supervision of $K$ under the
    attack set $\mathcal{A}$     
    if and only if $K$ is controllable and $P$-observable for
    $\mathcal{A}$;

  \item \label{en:supervisor-attacks} if $K$ is controllable and
    $P$-observable for $\mathcal{A}$, the map
    $f:\bigcup_{A \in \mathcal{A}}AP(L)\to \Sigma$ defined by
    \begin{align}
      \label{eq:CC_supervisor_attack}
      f(y) := \Sigma_{uc} \cup
      \Big\{
      \sigma \in \Sigma_c:
      \exists w \in K, A \in \mathcal{A} \st
      [y \in AP(w),~
      w\sigma \in K]
      \Big\},\quad
      \forall y \in \bigcup_{A \in \mathcal{A}}AP(L).
    \end{align}
    is a solution to the supervision of $K$ under the attack set
    $\mathcal{A}$. \hspace*{\fill} $\Box$
  \end{enumerate}
\end{theorem}

Theorem~\ref{th:observability-attacks} enable us to conclude that if
$K$ is controllable and $P$-observable for $\mathcal{A}$ then there
exists a Stackelberg equilibrium policy for the leader (supervisor)
that guarantees victory and Theorem~\ref{th:observability-attacks}
provides one such policy in
\eqref{eq:CC_supervisor_attack}. Conversely, if $K$ is not
controllable or not $P$-observable for $\mathcal{A}$, then there
exists no policy for the supervisory that can guarantee victory for an
adversary that knows the policy of the supervisory (in a Stackelberg
sense) and has full information about the game. However, this result
does not provide a solution to scenarios in which $K$ is not
controllable or not $P$-observable for $\mathcal{A}$, but the
supervisor does not advertise its policy or the adversary does not
have full information about the game. We conjecture that in such
scenarios, equilibrium policies for the supervisor will be stochastic
(either in a Nash or Stackelberg sense) and the value of the game will
be a probability of victory in the (open) set $(0,1)$.

\begin{remark}
  In the proof of Theorem~\ref{th:observability-attacks}, we will see
  that to conclude that $K$ is controllable, it suffices that
  $L_{f,A}^{\max}=K$ or $L_{f,A}^{\min}=K$ for some supervisory $f$
  and attack $A\in\mathcal{A}$. \hspace*{\fill} $\Box$
\end{remark}

As in the classic supervisory control problem, the fact
  that controllability of $K$ is a necessary condition for the
  existence of a supervisor is to be expected since, if
  $K \Sigma_{uc} \cap L$ were not a subset of $K$, there would exist a
  string $\bar w\in K \Sigma_{uc}\cap L$, $\bar w\notin K$ that should not be
  accepted (because $\bar w\notin K$) and yet $\bar w$ is the extension of
  a string that should be accepted (because it belongs to $K$)
  followed by an uncontrolled event in $\Sigma_{uc}$.  Once we properly
  parse the meaning of $P$-observability for an attack set $\mathcal{A}$,
  the necessity of this condition is easy to understand: As noted
  above, when this condition does not hold, it is possible to find two
  attacks $A, \bar A \in \mathcal{A}$ that would result in the same
  observation $y\in\Delta^*$ for two distinct strings
  $w, \bar w\in K$, and yet $w$ would transition to an element in $K$
  and $\bar w$ to an element outside $K$. The existence of such
  strings would mean that, upon observing $y$, the supervisor could
  not decide whether or not to enable the next transition. It turns
  out that these two conditions are also sufficient to guarantee that
  we can construct a $P$-supervisor $f$ for that guarantees that
  $L_{f,A}^{\min}=L_{f,A}^{\max}=K$ for all $A\in\mathcal{A}$ and
  Theorem~\ref{th:observability-attacks} provides such a
  supervisory. The following alternative characterization of
  observability under attacks will be used in the sufficiency proof.

\begin{proposition}
  \label{prop:Observability_equivalence}
  Suppose that prefix-closed $K \subset L$ is controllable.
  Then $K$ is $P$-observable for the set of attacks
  $\mathcal{A}$ if and only if
  for all $w, \bar w \in K$, $\sigma \in \Sigma_c$, $A,\bar A \in \mathcal{A}$,
  the following statement holds:
  \begin{align}
    \label{eq:observability_remove_cont}
    [AP(w) \cap \bar AP(\bar w) \neq \emptyset,\quad
    w\sigma \in K,\quad \bar w\sigma \in L]
    \quad \Rightarrow \quad \bar w\sigma \in K.
  \end{align}
\hspace*{\fill} $\Box$
\end{proposition}

\noindent\hspace{1em}{\itshape {\it Proof of Proposition \ref{prop:Observability_equivalence}:} }
  To prove this result, we use the necessary and sufficient condition
  \eqref{eq:observability_nonexist} for observability.

  \medskip

  ($\Rightarrow$) Suppose that $K$ is $P$-observable for
  $\mathcal{A}$, and suppose that $w, \bar w \in K$,
  $\sigma \in \Sigma_c$, $A,\bar A \in \mathcal{A}$ satisfy
  $AP(w) \cap \bar AP(\bar w) \neq \emptyset$, $w\sigma \in K$, and
  $\bar w\sigma \in L$.  Then \eqref{eq:observability_nonexist} leads to
  $\bar w\sigma \in K$.

  \medskip

  ($\Leftarrow$) Suppose that \eqref{eq:observability_remove_cont} holds for
  all $w, \bar w \in K$, $\sigma \in \Sigma_c$, $A,\bar A \in \mathcal{A}$.  Let
  $w, \bar w \in K$ satisfy $AP(w) \cap \bar AP(\bar w) \neq \emptyset$ for some
  $A,\bar A \in \mathcal{A}$.  Since $K$ is controllable, if
  $\sigma \in \Sigma_{uc}$ and $w\sigma, \bar w\sigma \in L$, then
  $w\sigma, \bar w\sigma \in K$.  Moreover, we see from
  \eqref{eq:observability_remove_cont} that, for all
  $\sigma \in \Sigma_c$, if $w\sigma \in K$ and $\bar w\sigma \in L$, then
  $\bar w\sigma \in K$.  Exchanging $w$ and $\bar w$, we also have if
  $\bar w\sigma \in K$ and $w\sigma \in L$, then $w\sigma \in K$.  Thus there does not exist
  $\sigma \in \Sigma$ such that
  [$w\sigma \in K,~ \bar w\sigma \in L \setminus K$] or
  [$w\sigma \in L \setminus K,~ \bar w\sigma \in K$].\hspace*{\fill} $\clubsuit$

\noindent\hspace{1em}{\itshape {\it Proof of Theorem \ref{th:observability-attacks}:} }
  To prove necessity in item~\ref{en:observability-attacks}, assume
  that there exists a valid $P$-supervisor $f$ such that
  $L_{f,A}^{\min}=L_{f,A}^{\max}=K$ for all $A\in\mathcal{A}$.  To prove
  that $K$ is controllable, we (only) use the fact that
  $K=L_{f,A}^{\min}$ and pick some word
  $\bar w \in K\Sigma_{uc}\cap L$.  Such a word must be of the form
  $\bar w=w\sigma\in L$ such that $w\in K=L_{f,A}^{\min}$ for all
  $A\in\mathcal{A}$ and $\sigma\in\Sigma_{uc}$. But then,
  \begin{align*}
    w\in L_{f,A}^{\min}, \quad w\sigma\in L, \quad \forall y\in A P(w),\; \sigma \in\Sigma_{uc}\subset f(y).
  \end{align*}
  The definition of $L_{f,A}^{\min}$ thus allows us to conclude that
  $\bar w=w\sigma\in L_{f,A}^{\min}=K$ for all
  $A\in\mathcal{A}$, which shows that
  $K\Sigma_{uc}\cap L\subset K$ and therefore $K$ is controllable.

  \medskip

  To prove that $K$ is $P$-observable, we use the fact that
  $K=L_{f,A}^{\min}=L_{f,A}^{\max}$, $\forall A\in\mathcal{A}$ and pick a
  pair of words $w,\bar w\in K$ such that
  \begin{align*}
    \exists A,\bar A \st A P(w)\cap \bar AP(\bar w)\neq \emptyset,
  \end{align*}
  and an arbitrary symbol $\sigma\in\Sigma$ such that
  $w\sigma,\bar w\sigma\in L$. If $w\sigma\in K=L_{f,A}^{\min}$,
  then by the definition of
  $L_{f,A}^{\min}$, we must have
  \begin{align*}
    w\in L_{f,A}^{\min}=K,\quad w\sigma\in L,\quad \forall y\in A P(w),\; \sigma\in g(y).
  \end{align*}
  Since $A P(w)\cap \bar AP(\bar w)\neq \emptyset$, we must then have
  \begin{align*}
    \bar w\in L_{f,\bar A}^{\max}=K,\quad \bar w\sigma\in L,\quad \exists y\in \bar A P(\bar w) \st \sigma\in f(y)
  \end{align*}
  and consequently $\bar w\sigma\in L_{f,\bar A}^{\max}=K$.  Alternatively, if $w\sigma\in L\setminus
  K$, then $w\sigma\notin K=L_{f,A}^{\max}$ and we must have
  \begin{align*}
    w\in L_{f,A}^{\max}=K,\quad w\sigma\in L,\quad \forall y\in A P(w),\; \sigma\notin f(y).
  \end{align*}
  Since $A P(w)\cap \bar AP(\bar w)\neq \emptyset$, we must then have
  \begin{align*}
    \bar w\in L_{f,\bar A}^{\min}=K,\quad \bar w\sigma\in L,\quad \exists y\in \bar A P(\bar w) \st \sigma\notin f(y).
  \end{align*}
  and consequently $\bar w\sigma\notin L_{f,\bar A}^{\min}=K$.  This shows that
  \eqref{eq:act-K-attacks} holds and therefore $K$ is $P$-observable
  for $\mathcal{A}$.

  \medskip

  To prove sufficiency in item \ref{en:observability-attacks} (and
  also the statement in item~\ref{en:supervisor-attacks}),
  pick the supervisor according to \eqref{eq:CC_supervisor_attack}.
  We prove by induction on the word length that the supervisor $f$ so
  defined satisfies $K=L_{f, A}^{\max}=L_{f,A}^{\min}$ for all
  $ A\in\mathcal{A}$. The basis of induction is the empty string
  $\epsilon$ that belongs to $L_{f,A}^{\max}$ because of the definition of
  this set and belongs to $K$ because this set is prefix-closed.

  \medskip

  Suppose now that $K$, $L_{f,A}^{\max}$, and $L_{f,A}^{\min}$ have
  exactly the same words of length $n\ge 0$, and pick a word
  $\bar w \sigma \in L_{f,A}^{\max}$ of length $n+1$. Since
  $\bar w\in L_{f,A}^{\max}$ has length $n$, we know by the induction
  hypothesis that $\bar w\in K$. On the other hand, since
  $\bar w \sigma\in L_{f,A}^{\max}$, we must have
  \begin{align*}
    \bar w\in L_{f,A}^{\max},\quad \bar w\sigma\in L,\quad \exists y\in A P(\bar w) \st \sigma\in f(y).
  \end{align*}
  If $\sigma \in \Sigma_{uc}$, then we see from controllability that
  $\bar w\sigma \in K$.  Let us next consider the case
  $\sigma \in \Sigma_c$.  By the definition of $f$, $\sigma\in f(y)$ must mean that
  \begin{align*}
    \exists w\in K,\bar A \in\mathcal{A}
    \quad
    \st
    \quad
    [y\in \bar AP(w),~w\sigma \in K].
  \end{align*}
  We therefore have
  \begin{align*}
    w,\bar w \in K,\quad
    AP(\bar w) \cap \bar AP(w) \neq \emptyset,\quad
    w\sigma \in K,\quad \bar w\sigma \in L.
  \end{align*}
  Proposition \ref{prop:Observability_equivalence} therefore shows
  that $\bar w\sigma \in K$.  This shows that any word of length $n+1$ in
  $L_{f,A}^{\max}$ also belongs to $K$. Since
  $L_{f,A}^{\min}\subset L_{f,A}^{\max}$, it follows that any word of length
  $n+1$ in $L_{f,A}^{\min}$ also belongs to $K$.


  \medskip

  Conversely, pick a word $\bar w \sigma \in K\subset L$ of length $n+1$. Since
  $K$ is prefix closed, $\bar w\in K$, and by the induction hypothesis
  that $\bar w\in L_{f,A}^{\min}$. To obtain
  $\bar w\sigma \in L_{f,A}^{\min}$, we need to show that
  \begin{align*}
    \bar w\in L_{f,A}^{\min},\quad \bar w\sigma\in L,\quad \forall y\in A P(\bar w),\; \sigma\in f(y).
  \end{align*}
  The first statement is a consequence of the induction hypothesis (as
  discussed above). The second statement is a consequence of the fact
  that $\bar w \sigma \in K\subset L$. As regards the third statement, if
  $\sigma \in \Sigma_{uc}$, then $\sigma \in f(y)$ for all
  $y\in A P(\bar w)$ by definition.  In order to show that if
  $\sigma \in \Sigma_{c}$, then $\sigma\in f(y)$ for all
  $y\in A P(\bar w)$, we need to prove that
  \begin{align}\label{eq:P-P-attacks}
    \forall y\in A P(\bar w),\quad
    \exists w\in K,~\bar A\in\mathcal{A} \st
    [y\in \bar AP(w),~w\sigma\in K].
  \end{align}
  This holds for the particular case $\bar A=A$, $w=\bar w \in K$.
  Thus any word of length $n+1$ in $K$ also belongs to
  $L_{f,A}^{\min}$ and hence also to
  $L_{f,A}^{\max}\supset L_{f,A}^{\min}$, which completes the induction
  step. 
  \hspace*{\fill} $\clubsuit$

\begin{remark}
  Similarly, we can show that 
  \begin{align}
    \label{eq:b_f_def}
    \bar f(y) := \Sigma_{uc} \cup
    \Big\{
    \sigma \in \Sigma_c:
    \forall w \in K, A \in \mathcal{A} \st
    [y \in AP(w),~
    w\sigma \in L]
    ~~
    \Rightarrow
    ~~
    w\sigma \in K
    \Big\},\quad
    \forall y \in \bigcup_{A \in \mathcal{A}}AP(L)
  \end{align}
  is also a solution to the supervision of $K$ under the attack set
  $\mathcal{A}$. The supervisor $f$ in \eqref{eq:CC_supervisor_attack} is
  said to be \emph{permissive}, while $\bar f$ in \eqref{eq:b_f_def}
  is \emph{anti-permissive} \cite{Yoo2002,Ushio2009}. \hspace*{\fill} $\Box$
\end{remark}

\section{Output-symbol attacks}\label{se:output-symbol-attacks}

Given a set of symbols $\alpha\subset\Delta$ in the observation alphabet, one can
define an observation attack $A_\alpha:\Delta^*\to 2^{\Delta^*}$ that maps to each
string $u\in\Delta^*$ the set of all strings $v\in \Delta^*$ that can be obtained
from $u$ by an arbitrary number of insertions or deletions of symbols
in $\alpha$. We say that $A_\alpha$ corresponds to an \emph{attack on the output
  symbols in $\alpha$.} In this context, it is convenient to also define
the corresponding \emph{$\alpha$-removal observation map}
$R_{\neg\alpha}:\Delta\to (\Delta \cup \{\epsilon\})$ by
\begin{align*}
  R_{\neg\alpha}(t) =
  \begin{cases}
    \epsilon & t \in \alpha \\
    t & t \not\in \alpha.
  \end{cases}
\end{align*}
The $\alpha$-removal observation map can be extended to a map defined for
strings of output symbols in the same way as observation maps $P$.
This $\alpha$-removal observation map allows us to define
$A_\alpha:\Delta^*\to 2^{\Delta^*}$ as follows:
\begin{align}\label{eq:A-R}
  A_\alpha(u)=\big\{v\in \Delta^*:R_{\neg\alpha}(u)=R_{\neg\alpha}(v)\big\}.
\end{align}
Note that the absence of attack $A_{\rm id}$ corresponds to an empty
set $\alpha=\emptyset$. This type of attacks are precisely the ones we found in
Example~\ref{ex:security}.

\subsection{Observability test}\label{se:output-symbol-attacks-test}

The next result shows that for attacks on output symbols, one can test
observability under attacks by checking regular observability (without
attacks) for an appropriate set of output maps. This means that the
observability tests developed for the non-attacked case
\cite{CassandrasLafortune2008,Tsitsiklis1989} can be used to determine
observability under output-symbol attacks.
\begin{theorem}\label{th:symbol-attack}
  For every nonempty prefix-closed set $K\subset L$ and attack set
  $\mathcal{A}=\{A_{\alpha_1},A_{\alpha_2},\dots,A_{\alpha_M}\}$ consisting of
  $M\ge 1$ observation attacks, $K$ is $P$-observable for the set of
  attacks $\mathcal{A}$ if and only if $K$ is
  $(R_{\neg\alpha}\circ P)$-observable (in the classical sense, i.e., without
  attacks) for every set $\alpha:= \alpha_i\cup\alpha_j$, $\forall i,j\in\{1,2,\dots,M\}$.
   \hspace*{\fill} $\Box$
\end{theorem}


  In essence, Theorem~\ref{th:symbol-attack} states that to have
  $P$-observability for a set of output-symbol attacks, we need to pick
  every possible pair of two attacks $A_{\alpha_i}$, $A_{\alpha_j}$ and ask
  whether we would have ``classical'' observability if we were to
  remove all symbols affected by the two attacks. This result can be
  counter-intuitive because even though we assume that we only have
  one attack in $\mathcal{A}$, we need to protect consider the effect of
  pairs of attacks. If we know which attack in $\mathcal{A}$ we had to
  face, then it would suffice to erase from the output all symbols
  corresponding to that particular attack. The problem is that we do
  not know which attack we are facing and this force us to have more
  ``redundancy'' in the sense that we need a stronger version of
  observability. In fact, we will see in
  Theorem~\ref{th:supervisor-symbol-attack}, that we can construct a
  supervisor to solve this problem precisely by removing more symbols
  than those an attacker could control and then checking for
  consistency across the decisions made by ``classical'' supervisors
  that operate on reduced sets of symbols.
  
  \medskip

The following result provides the key step needed to prove
Theorem~\ref{th:symbol-attack}.

\begin{lemma}\label{le:R-alpha}
  Given any two sets $\alpha_1,\alpha_2\subset\Delta$ and $\alpha:=\alpha_1\cup\alpha_2$, we have that
  \begin{align}
    \label{eq:z_to_A}
    (v,\bar v)\in\ker R_{\neg\alpha} \quad \Rightarrow \quad
    A_{\alpha_1}(v)\cap A_{\alpha_2}(\bar v)\neq\emptyset,
  \end{align}
  where $\ker R_{\neg \alpha} :=
  \{(v,\bar v) \in \Delta^* \times \Delta^*:
  R_{\neg \alpha}(v) = R_{\neg \alpha}(\bar v)
  \}$.
 \hspace*{\fill} $\Box$
\end{lemma}

\noindent\hspace{1em}{\itshape {\it Proof of Lemma \ref{le:R-alpha}:} }
  To prove this result, we must show that given two words
  $v,\bar v\in\Delta^*$ such that $R_{\neg\alpha}(v)=R_{\neg\alpha}(\bar v)$, there exists a third
  word $y\in\Delta^*$ that belongs both to $A_{\alpha_1}(v)$ and
  $A_{\alpha_2}(\bar v)$, and therefore
  \begin{align}
    \label{eq:R1_R2_eq}
    R_{\neg \alpha_1}(v)=R_{\neg \alpha_1}(y), \quad
    R_{\neg \alpha_2}(\bar v)=R_{\neg \alpha_2}(y).
  \end{align}
  The desired word $y$ can be constructed through the following steps:
  \begin{enumerate}
  \item Start with the word $y_1=R_{\neg \alpha_1}(v)$, which is obtained by
    removing from $v$ all symbols in $\alpha_1$.  Since
    $R_{\neg\alpha}=R_{\neg\alpha_2}\circ R_{\neg\alpha_1}$, we have that
    \begin{align*}
      R_{\neg \alpha_2}(y_1)
      = R_{\neg \alpha_2}(R_{\neg \alpha_1}(v))
      = R_{\neg \alpha}(v)
      = R_{\neg \alpha}(\bar v),
    \end{align*}
    and therefore the words $y_1$ and $\bar v$ still only differ by symbols
    in $\alpha_1$ and $\alpha_2$.

  \item Construct $y_2$ by adding to $y_1$ suitable symbols in
    $\alpha_1$ so that $R_{\neg \alpha_2}(y_2)=R_{\neg \alpha_2}(\bar v)$. This is possible because
    $y_1$ and $\bar v$ only differ by symbols in $\alpha_1$ and
    $\alpha_2$.  To get
    $R_{\neg \alpha_2}(y_2)=R_{\neg \alpha_2}(\bar v)$, we do not care about the symbols in
    $\alpha_2$ so we just have to insert into $y_1$ the symbols in
    $\alpha_1$ that appear in $\bar v$ (at the right locations).

  \item By construction,
    $R_{\neg \alpha_1}(y_2) =R_{\neg \alpha_1}(y_1) =y_1= R_{\neg \alpha_1}(v)$, and hence the original
    $v$ and $y_2$ only differ by symbols in $\alpha_1$. Since
    $R_{\neg\alpha_2}(y_2)=R_{\neg\alpha_2}(\bar v)$, it follows that $\bar v$ and
    $y_2$ only differ by symbols in $\alpha_2$. We therefore conclude that
    $y:= y_2$ indeed satisfies \eqref{eq:R1_R2_eq} and hence belongs to both
    $A_{\alpha_1}(v)$ and $A_{\alpha_2}(\bar v)$. \hspace*{\fill} $\clubsuit$
  \end{enumerate}

\noindent\hspace{1em}{\itshape {\it Proof of Theorem \ref{th:symbol-attack}:} }
  By definition, $K$ is $P$-observable for the set of attacks
  $\mathcal{A}$ if and only if
  \begin{align*}
    R_{A_{\alpha_i},A_{\alpha_j}}\subset act_{K\subset L},\qquad \forall i,j\in\{1,2,\dots,M\}.
  \end{align*}
  Also by definition, $K$ is $(R_{\neg\alpha}\circ P)$-observable for the set
  $\alpha:= \alpha_i\cup\alpha_j$ if and only if
  \begin{align*}
    \ker (R_{\neg\alpha}\circ P) \subset act_{K\subset L}.
  \end{align*}
  To prove the results, it therefore suffices to show that, for every
  $i,j\in\{1,2,\dots,M\}$, we have that
  \begin{align*}
    R_{A_{\alpha_i},A_{\alpha_j}}=\ker (R_{\neg\alpha}\circ P), \qquad \alpha:= \alpha_i\cup\alpha_j.
  \end{align*}
  To show that this equality holds, first pick a pair
  $(w,\bar w)\in R_{A_{\alpha_i},A_{\alpha_j}}$, which means that
  there exists a word $y\in\Delta^*$ that belongs both to
  $A_{\alpha_i} P(w)$ and $A_{\alpha_j}P(\bar w)$, and therefore
  \begin{align*}
    &y\in A_{\alpha_i} P(w) \quad \Leftrightarrow \quad R_{\neg\alpha_i}\big(P(w)\big)=R_{\neg\alpha_i}(y)\\
    &y\in A_{\alpha_j} P(\bar w) \quad \Leftrightarrow \quad R_{\neg\alpha_j}\big(P(\bar w)\big)=R_{\neg\alpha_j}(y)
  \end{align*}
  But then, since $\alpha:= \alpha_i\cup \alpha_j$, we have that $R_{\neg\alpha}=R_{\neg\alpha_j}\circ
  R_{\neg\alpha_i}=R_{\neg\alpha_i}\circ R_{\neg\alpha_j}$, and consequently
  \begin{multline*}
    R_{\neg\alpha}\big(P(w)\big)
    =R_{\neg\alpha_j}\big(R_{\neg\alpha_i}\big(P(w)\big)\big)
    =R_{\neg\alpha_j}\big(R_{\neg\alpha_i}(y)\big)\\
    =R_{\neg\alpha_i}\big(R_{\neg\alpha_j}(y)\big)
    =R_{\neg\alpha_i}\big(R_{\neg\alpha_j}\big(P(\bar w)\big)\big)
    =R_{\neg\alpha}\big(P(\bar w)\big).
  \end{multline*}
  We have thus shown that
  $R_{A_{\alpha_i},A_{\alpha_j}}\subset \ker (R_{\neg\alpha}\circ P)$. To prove the reverse
  inclusion, pick a pair $(w,\bar w)\in \ker (R_{\neg\alpha}\circ P)$, which means that
  \begin{align*}
    R_{\neg\alpha}\big(P(w)\big)=R_{\neg\alpha}\big(P(\bar w)\big),
  \end{align*}
  and therefore $\big(P(w),P(\bar w)\big)\in\ker R_{\neg\alpha}$.
  In conjunction with
  Lemma~\ref{le:R-alpha}, this gives
  \begin{align*}
    A_{\alpha_i}\big(P(w)\big)\cap
    A_{\alpha_j}\big(P(\bar w)\big)\neq \emptyset.
  \end{align*}
  Therefore we have $(w,\bar w)\in R_{A_{\alpha_i},A_{\alpha_j}}$. This shows that
  $\ker (R_{\neg\alpha}\circ P)\subset R_{A_{\alpha_i},A_{\alpha_j}}$, which concludes
  the proof.
 \hspace*{\fill} $\clubsuit$

\begin{example}[Multi-layer cyber attack to a computer system
  (cont.)]\label{ex:security-test}
  In Example~\ref{ex:security} we considered $M$ attacks
  $A_\text{exploit\; i}$, $i\in\{1,\dots,M\}$, each corresponding to an
  adversary having compromised the security system that logs the
  occurrence of ``exploit i'', allowing it to arbitrarily
  insert/remove the output symbols ``exploit i'' into/from the
  observation string. Each of these is an output-symbol attack
  $A_\alpha$, with the set $\alpha$ including a single output symbol ``exploit
  i''. Using the $A_\alpha$-notation introduced above, we can thus write
  the attack set $\mathcal{A}_{\rm log-attacks}$ in \eqref{eq:log-attacks}
  as.
  \begin{align*}
    \mathcal{A}_{\rm log-attacks}:=\big\{A_{ \emptyset },A_{\{\text{exploit 1}\}},A_{\{\text{exploit
        2}\}},\dots,A_{\{\text{exploit M}\}} \big\}.
  \end{align*}
  Theorems~\ref{th:observability-attacks} and~\ref{th:symbol-attack}
  can be used to confirm our previous assertion that there exists a
  solution to the supervision of the language $K_{\rm safe}$ under the
  attack set $\mathcal{A}_{\rm log-attacks}$ if and only if $M\ge 3$. This
  is because $K_{\rm safe}$ is $P$-observable for the attack set
  $\mathcal{A}_{\rm log-attacks}$ if and only if $M\ge 3$:

  \begin{enumerate}
  \item For $M=1$, $K_{\rm safe}$ is not $P$-observable for the attack
    set
    $\mathcal{A}_{\rm log-attacks}:=\big\{A_{ \emptyset },A_{\{\text{exploit
        1}\}}\big\}$ because $K_{\rm safe}$ is not
    $(R_{\neg\{\text{exploit 1}\}}\circ P)$-observable. This is
    straightforward to conclude from Theorem~\ref{th:symbol-attack}
    because if we remove from the output strings all symbols ``exploit
    1'', the supervisor cannot distinguish between the ``clean'' and
    ``s1'' states.
  
  \item For $M=2$, $K_{\rm safe}$ is still not $P$-observable for the
    attack set
    $\mathcal{A}_{\rm log-attacks}:=\big\{A_{ \emptyset },A_{\{\text{exploit
        1}\}},A_{\{\text{exploit 2}\}}\big\}$ because $K_{\rm safe}$
    is not
    $(R_{\neg\{\text{exploit 1},\text{exploit 2}\}}\circ
    P)$-observable. Again, this is straightforward to conclude from
    Theorems~\ref{th:symbol-attack} because if we remove from the
    output strings all symbols ``exploit 1'' and ``exploit 2'', the
    supervisor cannot distinguish between the ``clean'', ``s1'',
    ``s2'' states.

    \medskip

    This result may seem counter-intuitive because none of the
    attacks in $\mathcal{A}_{\rm log-attacks}$ is actually able to
    insert/remove both symbols ``exploit 1'' and ``exploit 2'' and yet
    the test involves a system where we removed both symbols. To
    understand this apparent paradox, suppose that the supervisor
    observes the output string
    \begin{align*}
      \{\text{grant access},\text{exploit 1}\}.
    \end{align*}
    This output string could have resulted from two scenarios:
    \begin{enumerate}
    \item The event sequence is
      \begin{align*}
        \{\text{grant access},\text{exploit 1},\text{exploit 2}\},
      \end{align*}
      but an output-symbol attack $A_{\{\text{exploit 2}\}}$ erased
      the symbol ``exploit 2''.
    \item The event sequence is
      \begin{align*}
        \{\text{grant access}\},
      \end{align*}
      but an output-symbol attack $A_{\{\text{exploit 1}\}}$ inserted
      the symbol ``exploit 1''.
    \end{enumerate}
    \smallskip The existence of these two options is problematic because
    in the former case ``grant access'' must be disabled, whereas in
    latter case it must be enabled.
    
  \item For $M\ge 3$, $K_{\rm safe}$ becomes $P$-observable for the
    attack set
    \begin{align*}
      \mathcal{A}_{\rm log-attacks}:=\big\{A_{ \emptyset },A_{\{\text{exploit
          1}\}},\dots,A_{\{\text{exploit M}\}}\big\}
    \end{align*}
    because $K_{\rm safe}$ is
    $(R_{\neg\{\text{exploit i},\text{exploit i}\}}\circ P)$-observable, for
    every $i,j\in\{1,2,\dots,M\}$. Again, this is straightforward to
    conclude from Theorems~\ref{th:symbol-attack} because even if we
    remove from the output strings all pairs of symbols ``exploit i''
    and ``exploit j'', there will still remain a third symbol
    ``exploit k'' $k\neq i, k\neq j$ in the output string that allows the
    supervisor to deduce a transition out of the ``clean''
    state. \hspace*{\fill} $\Box$
\end{enumerate}
\end{example}

\begin{remark}[complexity]
  As noted in Remark~\ref{re:expansion}, it is often possible to
  reduce the problem of finding a solution to the supervision under an
  attack set $\mathcal{A}$ to a classical supervisory control problem
  without attacks, at the expense of having to expand the state of the
  original automaton, essentially by replicating each state
  $M:=|\mathcal{A}|$ times.
%
%
  For problems in which the complexity of testing observability is
  linear in the number of states of the original automaton, this
  reduction may be computationally more efficient than using the
  result in Theorem~\ref{th:symbol-attack}, which would require
  testing the observability of $\binom{M}{2}=O(M^2)$ systems.

  \medskip

  The simple structure of Example~\ref{ex:security-test} allowed us to
  prove observability for every $M\ge3$ through a formal
  argument. However, more complicated networks of exploits typically
  require the use of computational tests for observability.  If we had
  to perform a computational test for Example~\ref{ex:security-test},
  Theorem~\ref{th:symbol-attack} would require testing the
  observability of $\binom{M}{2}=(M^2+M)/2$ distinct systems, each with
  $O(M)$ states and transitions. For this problem, one can show that
  each of these tests would have a worst-case computational complexity
  $O(M^3)$, using the algorithm in \cite[Section
  3.7.3]{CassandrasLafortune2008}. This would lead to a worst-case
  complexity $O(M^5)$.  Alternatively, the expansion described in
  Remark~\ref{re:expansion} would lead to a single observability test
  for a system with $O(M^2)$ states and transitions. Even with more
  states, the observability test for this system based on the
  algorithm in \cite[Section 3.7.3]{CassandrasLafortune2008} would
  still only have a worst-case complexity $O(M^4)$.
  However, while sometimes attractive to \emph{test observability}, we
  shall see shortly that the expansion described in
  Remark~\ref{re:expansion} generally results in much higher
  worst-case complexity for the \emph{design of the supervisor}.
\hspace*{\fill} $\Box$
\end{remark}

\subsection{Supervisor design}\label{se:output-symbol-attacks-supervisor}

The following result shows that for attacks on output symbols, one can
construct a solution to the supervision of $K\subset L$ under attacks by
appropriately combining a set of classical supervisors, each designed
for an appropriately defined output map without attacks. This allow us
to reuse classical supervisor-design approaches and tools
\cite{FengWonham2006,CassandrasLafortune2008,DESUMA2014}.

\begin{theorem}\label{th:supervisor-symbol-attack}
  For a given nonempty prefix-closed set $K\subset L$ and attack set
  $\mathcal{A}=\{A_{\alpha_1},A_{\alpha_2},\dots,A_{\alpha_M}\}$ consisting of
  $M\ge 1$ observation attacks, assume that $K$ is controllable and
  $P$-observable for $\mathcal{A}$. In view of
  Theorems~\ref{th:observability-attacks} and~\ref{th:symbol-attack},
  this means that, for every $i,j\in\{1,2,\dots,M\}$ there exists a
  valid supervisor $f_{i j}$ that generate the language $K$ for the
  observation map $P_{i j} := R_{\neg(\alpha_i \cup \alpha_j)}\circ P$ (in the classical
  sense, i.e., without attacks). In this case, the following
  supervisor is a solution to the supervision of $K$ under the attack
  set $\mathcal{A}$:
  \begin{align*}
    f(y)&:=\bigcup_{i=1}^M f^i(y), &
    f^i (y) &:=  \bigcap_{j=1}^M \tilde f_{ij}(y), &
    \forall y \in  \bigcup_{i=1}^M A_{\alpha_i}P(L),
  \end{align*}
  where
  \begin{align}\label{eq:tilde-f}
    \tilde f_{ij}(y) :=  
    \begin{cases}
      f_{ij}(R_{\neg(\alpha_i \cup \alpha_j)} (y)) & y \in A_{\alpha_i}P(L) \cup A_{\alpha_j}P(L) \\
      \Sigma_{uc} & y \not\in A_{\alpha_i}P(L) \cup A_{\alpha_j}P(L).
    \end{cases}
  \end{align}
\hspace*{\fill} $\Box$
\end{theorem}  

\begin{remark}
  As we shall see in the proof of
  Theorem~\ref{th:supervisor-symbol-attack}, the lower branch of
  \eqref{eq:tilde-f} can be set equal to any subset of $\Sigma$ (possibly
  $y$-dependent) that contains $\Sigma_{uc}$.  \hspace*{\fill} $\Box$
\end{remark}

\noindent\hspace{1em}{\itshape {\it Proof of Theorem \ref{th:supervisor-symbol-attack}:} }
  Since every $f_{i j}$ is a valid supervisor, all the sets $f_{i j}(y)$
  and $\tilde f_{i j}(y)$ contain $\Sigma_{uc}$ and consequently so does
  $f(y)$, which proves that $f$ is a valid supervisor. To complete the
  proof, it thus suffices to show that
  \begin{align*}
    L_{f,A_{\alpha_k}}^{\min} \supset K \supset  L_{f,A_{\alpha_k}}^{\max}, \quad \forall A_{\alpha_k}\in\mathcal{A}.
  \end{align*}
  First we prove $L_{f,A_{\alpha_k}}^{\min} \supset K$ by showing that every word
  $w$ in $K$ also belongs to $L_{f,A_{\alpha_k}}^{\min}$. We do this prove
  by induction on word length $n$. The basis of induction is trivially
  true because the empty string belongs to $L_{f,A_{\alpha_k}}^{\min}$ by
  construction. Suppose now that we have established that every word
  $w \in K$ of length $n$ also belongs to $L_{f,A_{\alpha_k}}^{\min}$ and
  take an arbitrary word $w\sigma \in K$ of length $n+1$.
  Since $K$ is prefix closed, we have that $w\in K$ and by the induction
  hypothesis that $w\in L_{f,A_{\alpha_k}}^{\min}$.  In view of the definition of
  minimal language, to prove that
  $w\sigma\in L_{f,A_{\alpha_k}}^{\min}$ we need to show that
  $\sigma\in f(y)$, $\forall y\in A_{\alpha_k} P(w)$. To accomplish this, we pick an
  arbitrary string $y\in A_{\alpha_k} P(w)$ and note that, because every
  $f_{k j}$ generates the language $K$ for the observation map
  $P_{k j} := R_{\neg(\alpha_k \cup \alpha_j)}\circ P$, we have that
  \begin{align}
    \label{eq:Lf_iff}
    w\sigma \in K \quad
    \Leftrightarrow \quad
    w \in K,\;w\sigma \in L,\;\sigma \in f_{kj}\big(R_{\neg(\alpha_k \cup \alpha_j)}\big(P(w)\big)\big),
    \quad \forall j\in\{1,\dots,M\}.
  \end{align}
  Since $y\in A_{\alpha_k} P(w)$, we conclude from \eqref{eq:A-R} that
  \begin{align}
    \label{eq:fRcond}
    R_{\neg(\alpha_k \cup \alpha_j)} (y) = R_{\neg(\alpha_k \cup \alpha_j)}\big(P(w)\big).
  \end{align}
  Moreover, we now that $w\sigma \in K$, therefore \eqref{eq:Lf_iff}
  and \eqref{eq:fRcond} imply that
  \begin{align*}
    \sigma \in f_{kj}\big(R_{\neg(\alpha_k \cup \alpha_j)}\big(P(w)\big)\big)
    = f_{kj}\big(R_{\neg(\alpha_k \cup \alpha_j)}(y)\big)
    = \tilde f_{kj}(y),
    \quad \forall j\in\{1,\dots,M\}.
  \end{align*}
  We have thus shown that
  \begin{align*}
    \sigma \in f^k(y) = \bigcap_{j=1}^M \tilde f_{kj}(y)
    \subset f(y).
  \end{align*}
  This confirms that $w\sigma\in L_{f,A_{\alpha_k}}^{\min}$, which completes the
  induction step.

  \medskip

  We show next that $L_{f,A_{\alpha_k}}^{\max} \subset K$, also using an
  induction argument. As before, the basis of induction is trivial so
  to establish the induction step, we suppose that we have established
  that every word $w \in L_{f,A_{\alpha_k}}^{\max}$ of length $n$ also
  belongs to $K$ and take an arbitrary word
  $w\sigma \in L_{f,A_{\alpha_k}}^{\max}$ of length $n+1$. By the induction
  hypothesis, we have that $w\in K$ and to prove the induction step, we
  need to show that $w\sigma\in K$. From the construction of
  $L_{f,A_{\alpha_k}}^{\max}$, we known that
  $w\sigma \in L_{f,A_{\alpha_k}}^{\max}$ implies that
  $w \in L_{f,A_{\alpha_k}}^{\max}$ and
  \begin{align*}
    \exists y \in A_{\alpha_k}P(w)\quad \text{s.t.}\quad
    \sigma \in f(y) := \bigcup_{i=1}^M f^i(y)
  \end{align*}
  and therefore
  \begin{align*}
    \exists y \in A_{\alpha_k}P(w),\; i\in\{1,\dots,M\}\quad \text{s.t.}\quad
    \sigma \in f^i(y):=\bigcap_{j=1}^M \tilde f_{ij}(y)
    \subset \tilde f_{kj}(y)=f_{kj}\big(R_{\neg(\alpha_k \cup \alpha_j)} (y)\big).
  \end{align*}
  Also here \eqref{eq:fRcond} holds and we conclude that
  \begin{align*}
    \sigma\in f_{kj}\big(R_{\neg(\alpha_k \cup \alpha_j)} (y)\big)=f_{kj}\big(R_{\neg(\alpha_k \cup \alpha_j)}\big(P(w)\big)\big).
  \end{align*}
  Again, because $f_{k j}$ generates the language $K$ for the
  observation map $P_{k j} := R_{\neg(\alpha_k \cup \alpha_j)}\circ P$, the equivalence
  in \eqref{eq:Lf_iff} holds. Since we have shown that all the
  conditions in the right-hand side of \eqref{eq:Lf_iff} hold, we
  conclude that $w\sigma \in K$, which completes the proof of the induction
  step. \hspace*{\fill} $\clubsuit$

\begin{example}[Multi-layer cyber attack to a computer system
  (cont.)]\label{ex:security-supervisor}  To construct a
  supervisor for the problem described in
  Examples~\ref{ex:security}-\ref{ex:security-test} using
  Theorem~\ref{th:supervisor-symbol-attack} we need valid supervisors
  $f_{i j}$ that generate the language $K_{\rm safe}$ for the
  observation maps $P_{i j} := R_{\neg(\alpha_i \cup \alpha_j)}\circ P$, $\forall
  \alpha_i,\alpha_j\in\mathcal{A}_{\rm log-attacks}$. For $M\ge3$, all these supervisors
  can be the same and defined by
  \begin{align*}
    f_{ij}(y)=
    \begin{cases}
      \Sigma_{uc} \cup \{\text{grant access}\} & \text{$T(y)$ does not contain any ``exploit'' symbol}\\
      \Sigma_{uc} \cup \{\text{deauthorize}\} & \text{$T(y)$ contains some ``exploit'' symbols},
    \end{cases}
  \end{align*}
  where $T(y)$ denotes the last output symbols in the string $y$,
  starting right after the last ``grant access'' or ``deauthorize''
  symbols, or from the beginning of $y$ if it does not contain such
  symbols. Essentially $T(y)$ contains the output symbols starting from
  the last time that the system was in the ``clean'' state. These
  supervisors enforce the language $K_{\rm safe}$ because if $T(y)$
  does not contain any ``exploit'' symbol for the output map
  $P_{i j}:= R_{\neg(\alpha_i \cup \alpha_j)}\circ P$, then we must be in the ``clean''
  state (in case $P_{i j} $ did not remove any ``exploit'' symbol), in
  the ``s1'' state (in case $P_{i j}$ removed one ``exploit'' symbols),
  or in the ``s2'' state (in case $P_{i j}$ removed two ``exploit''
  symbols). In either case, it is okay to enable the ``grant access''
  transition and disable the ``deauthorize'' transition. On the other
  hand, if $T(y)$ does contain any ``exploit'' symbol, then we must be
  in one of the ``si'' states and it is safe to enable the
  ``deauthorize'' transition and disable the ``grant access''
  transition.

  \medskip

  For these supervisors $f_{i j}$, we have
  \begin{align*}
    \tilde f_{ij}(y)
    &=\begin{cases}
      \Sigma_{uc} \cup \{\text{grant access}\} & \text{$T(y)$ does not contain
        any ``exploit'' symbol other than those in $\alpha_i \cup \alpha_j$,}\\
      \Sigma_{uc} \cup \{\text{deauthorize}\} & \text{$T(y)$ contains some
        ``exploit'' symbols other those in $\alpha_i \cup \alpha_j$,}
    \end{cases}\\
    f^i (y) 
    &=\begin{cases}
      \Sigma_{uc} \cup \{\text{grant access}\} & \text{$T(y)$ does not contain
        any ``exploit'' symbol, other than that in $\alpha_i$,}\\
      \Sigma_{uc}   & \text{$T(y)$ contains one or more copies of the same
        ``exploit'' symbol, not in $\alpha_i$,}\\
      \Sigma_{uc} \cup \{\text{deauthorize}\}  & \text{$T(y)$ contains one or
        more copies of two distinct ``exploit'' symbols not in $\alpha_i$,}
    \end{cases}\\
    f(y)
    &=\begin{cases}
      \Sigma_{uc} \cup \{\text{grant access}\} & \text{$T(y)$ does not contain
        any ``exploit'' symbol,}\\
      \Sigma_{uc} \cup \{\text{grant access}\} & \text{$T(y)$ contains one or
        more copies of a single
        ``exploit'' symbol,}\\
      \Sigma_{uc} \cup \{\text{deauthorize}\}  & \text{$T(y)$ contains one or
        more copies of two distinct ``exploit'' symbols.}
    \end{cases}
  \end{align*}
  The implementation of this supervisor could be done with a finite
  state machine with $M+1$ symbols, one state would correspond to ``no
  exploit symbol observed in $T(y)$'' and the remaining $M$ states
  would be use to memorize the index of the 1st exploit symbol
  observed in $T(y)$. \hspace*{\fill} $\Box$
\end{example}

\begin{remark}[Complexity]
  While testing observability has polynomial complexity in the number
  of states of the \emph{plant automaton} $G$ representing the
  original language $L$ and on the number of states of the
  \emph{specification automaton} $G_K$ representing the desired
  language $K$, the number of states of the supervisor is typically
  exponential in the number of states of the plant automaton
  \cite{Tsitsiklis1989,CassandrasLafortune2008}. To apply
  Theorem~\ref{th:supervisor-symbol-attack}, we thus need to design
  $O(|\mathcal{A}|^2)$ supervisors, each with worst-case exponential
  complexity in the number of states of $G$, leading to a complexity
  of $O(|\mathcal{A}|^2 e^{|X|})$, where $|X|$ denotes the number of states
  of $G$.  If we were to use instead the reduction to a supervisory
  control problem without attacks discussed in
  Remark~\ref{re:expansion}, we would need to solve a single
  supervisor design problem. However, this problem would have
  worst-case exponential complexity on the number of states of an
  expanded automaton, leading to a much worse worst-case complexity of
  $O(e^{|\mathcal{A}| |X|})$. \hspace*{\fill} $\Box$
\end{remark}

\section{Conclusion}

Motivated by computer security applications, we introduced a
multi-adversary version of the DESs supervisory control problem, where
the supervisor is asked to enforce a desired language based on
measurements that have been corrupted by one of several potential
opponents, not know which is the actual opponent. For the particular
case of output-symbol attacks, we have show that testing observability
and constructing a supervisor that solves this multi-adversary problem
can be done using tools developed for the classical supervisory
control problem, without expanding the state-space of the original
plant automaton.

\medskip

Important direction for future research include the development of
efficient computational techniques for adversaries that use attacks
more general than output-symbol attacks and investigating distributed
solutions to the problem considered here. The latter is especially
interesting for scenarios with distributed cyber-security
sensors. As noted above, this paper also does not
  provide results for scenarios in which the controllability and
  observability conditions do not hold. We conjecture that, in the
  absence of observability under attacks, saddle-point policies for
  the attacker/supervisory that maximize/minimize the probability that
  the supervisor will make a mistake will be mixed, but we do not know
  of efficient algorithms to compute such policies. 


  



\bibliographystyle{spmpsci}      

\end{document}